\documentclass[12pt]{article}
\usepackage{mydefs} 

\title{What Juris Hartmanis taught me about Reductions}
\author{Neil Immerman\\
  \url{https://people.cs.umass.edu/~immerman/}}
\begin{document}
\maketitle

\begin{abstract}
I was a student of Juris Hartmanis at Cornell in the late 1970's.  He believed that there was great
potential in studying restricted reductions.  I describe here some of his influences on me and, in
particular, how his insights concerning reductions helped me to prove that nondeterministic space is
closed under complementation.
\end{abstract}

\section{Introduction}

I arrived in Ithaca in August 1975, a twenty-one-year old PhD student in the
Cornell Math Department. My plan was to study recursive function theory, a branch of logic that
examines which mathematical problems are solvable, in principle, by  computer.
Yet, the most interesting -- and also the most challenging -- course I took that first fall was CS 681, the
graduate algorithms course. This was taught by John Hopcroft using the first modern algorithms text
\cite{AHU}.  I knew how to program, having worked the previous year
as a software engineer for GTE Sylvania,  but Hopcroft introduced me to a new world, and he taught
me a new way to think.

I promptly enrolled in the next semester's course, CS 682, graduate theory of computation, taught by
Juris Hartmanis. If Hopcroft challenged my thinking of computer science, Hartmanis completely
transformed my thinking about computation.  Had Hartmanis known how dramatically his course had
changed me that semester, he would have smiled and said, ``Wow!''
It is no exaggeration to say that in introducing me to the subject that I have studied ever since,
this course changed my life.

With his colleague Dick Stearns, Hartmanis had recently helped create the
field of computational complexity. The picture Hartmanis drew in the theory of computation course was of
a new field built on clear and beautiful mathematical principles. Hartmanis and Stearns had worked
out some of the relationships between the important new computational classes they defined. However,
unlike the more mature field of recursive function theory, in 1976, many of the basic and
fundamental questions in computational complexity were 
open. This was particularly exciting to me. No one realized at the time how hard these problems
actually were.

\section{Reductions}

The story I tell here focuses on reductions, a favorite tool of Hartmanis.  I will explain how the
understanding of reductions I gained from Hartmanis helped me give a simple solution to a
problem in computational complexity that had eluded reseachers for twenty-five years.

I assume that the reader is familiar with standard definitions in computational complexity, such as
for reductions and completeness \cite{Sipser}.  As I go further into definitions from
Descriptive Complexity, I hope the reader will enjoy just reading on and seeing the general ideas.
Additional definitions and background can be found in \cite{book}.

A reduction is the way we compare the computational complexity of two problems, $A$ and $B$.
We write $A \leq B$ to mean that there is a reduction from $A$ to $B$.
The problems $A$ and $B$ are decision problems, meaning on input $w$, we answer ``1'' if $w\in A$
and ``0'' otherwise.  A reduction from $A$ to $B$ is a transformation, $t$, taking any such input
$w$, to an input $t(w)$ to the decision problem $B$ in such a way that the answer to the question,
``Is $t(w) \in B$?'' is always exactly the same as the answer to the question, ``Is $w \in A$?''.
Furthermore, the mapping $t$ must have much lower computational complexity than the complexity of
the problems $A$ and $B$.

When $A \leq B$, it follows that
the computational complexity of $A$ is less than or equal to the computational complexity of $B$.
Richard Ladner proved that this less-than-or-equal-to relation on the computational complexity of problems is dense, i.e.,
if $A < B$, then there is an intermediate problem, $I$, and  $A < I < B$ \cite{Ladner}.
Thomas Schaefer pointed out that for many natural classes of problems, however, these intermediate problems
do not occur.  Instead, there is a dichotomy, that is, every problem in such a class is either NP
complete, i.e., a hardest problem in NP, or it is easy, i.e., in P \cite{Schaefer}.  Schaefer's
paper has lead to a great deal of important and fascinating explorations of such complexity
dichotomies, including the recent proof of the Feder-Vardi dichotomy conjecture for all constraint
satisfaction problems over finite domains \cite{FV99,Bulatov,Zhuk}.

Hartmanis was interested in the fact that complete problems tend to remain complete via
surprisingly weak reductions.  Steve Cook proved that $\sat$ is NP complete via ptime Turing
reductions \cite{Cook}.  When Dick Karp produced many other important NP-complete problems, he used
ptime, many-one reductions \cite{Karp}. Jones showed that these problems stay complete via logspace
reductions \cite{Jones}.  Hartmanis, Steve Mahaney and I showed that one-way logspace reductions
suffice \cite{HIM}; in this model the transducer reads its input once from left to right.

When I introduced Descriptive Complexity, I was pleased to see that first-order reductions -- which
are the natural way to translate one logical problem to another -- preserve completeness
properties for all natural complete problems and all complexity classes that I examined
\cite{book}.  Les Valiant had defined a reduction that is even weaker than first-order reductions,
however, and yet still preserves natural complete problems \cite{Valiant}.  Valiant called his reductions
projections: they are non-uniform and can perform no computations except negation (Def.\ \ref{projection de}).

I was pleased to observe that there is a natural restriction of first-order reductions making them
(first-order uniform) projections.  I call these first-order projections (fops) and observe that
natural problems stay complete via fops \cite{book}.  The weakness of fops is their strength as a
measuring tool; fops are the electron microscopes of computational complexity.

In the next section, I sketch relevant information from Descriptive Complexity.
In the following section I define uniform versions of Valiant's projections, namely first-order
projections (fop) and 
quantifier-free projections (qfp). The standard complexity problems for important complexity classes
remain complete for fops and sometimes even for qfps.  I will explain why this is remarkable and how
it fits in well with the research directions that Hartmanis pursued with reductions.

Finally, I
will show how the completness of graph reachability ($\reach$) for NL via qfps, led directly to my
proof that nondeterministic space is closed under complementation.  This was a surprising result and
a proof that Hartmanis was particulary pleased with.

\section{Descriptive Complexity}

Descriptive Complexity began with Fagin's Theorem, which gives a logical characterization
of $\np$.  We think of decision problems as sets of finite structures of a given finite
vocabulary.  A structure, $\cA=(A, R_1, \ldots, R_k)$ consists of a universe and a set of relations
over that universe, for example, a graph, $G=(V,E)$, has universe, $V$, the set of vertices, and the
binary edge relation, $E$.  Two examples of decision problems are  $\maj$, the set of binary strings, a majority of whose bits are ``1''
and $\Col$, the set of finite 3-colorable graphs.  Inputs to computers
are always binary strings, but we can think of them as encodings of other objects, e.g.,
structures encoding binary strings, graphs, etc.  In this framework, the input is a finite 
structure.  Thus, Descriptive Complexity is aligned with a part of
Mathematical Logic called Finite Model Theory in which only finite structures are considered.

\subsection{Ordered Structures and Numeric Formulas} 

When an input, $w$, is presented to a Turing Machine (TM) or other computer, it is in the form of a
binary string, i.e., $w \in \set{0,1}^n$ is a binary string of length $n$.  The number $n$ is always
used to denote the length of the input.  The input $w$ may be encoded or thought of as a 
structure, $\A_w$, which is also ordered, because a binary string is a finite ordered sequence of
bits.  Without the ordering, $\A_w$ would just be a bag of bits, losing almost all of its meaning.

However, suppose the input is a directed graph, $G= (V,E)$, i.e., a set of vertices, $V$, some
ordered pairs of which have an edge between them, $E \subseteq V \times V$.  It may seem easy to
understand the mathematical concept of an unordered graph.  However, no matter how you encode a
graph as a binary string, $w_G$, for example this could be a $\abs{V}^2$-bit adjacency matrix for
$G$, the string imposes an ordering on $V$:  there is an explicitly named first vertex, second
vertex, and so on.  This is unavoidable.  Thus, for all structures considered in descriptive complexity,
the logics considered have access to an explicit total ordering on the universe.  

Here we
assume throughout that the structures under consideration are ordered.  We have
constant symbols $\mn,\mx$ referring the the first and last elements in the ordering and 
relation symbol $\suc$ referring to the successor relation of the ordering.  The symbols $\mn, \mx,
\suc$ are called  {\em numeric} symbols and any formula using only logical or numeric symbols and no
input symbols, such as $E$, are called {\em numeric formulas} because their meaning only depends on
the size of the input structure, not which structure of that size it is.  It is crucial that the
languages under consideration have access to the total ordering of the universe, in this case via
the above numeric symbols.

\subsection{Second-Order Logic and Fagin's Theorem}

Second-order logic consists of first-order logic plus relation variables over
which we may quantify. Any second-order formula may be transformed into an equivalent formula with
all second-order quantifiers in front.  If all these second-order quantifiers are existential, then
we have a second-order existential formula.  Let $\soe$ be the set of decision problems expressible
by second-order existential sentences.

As an example, let $\Col$ be the NP-complete problem of deciding whether the vertices of an input
graph can be colored with three colors in a way that no two adjacent vertices have the same color.
How easy is it to express $\Col$ in $\soe$?  Very easy!  Let 
$R,Y,$ and $B$ be unary relation variables.  These represent three subsets of the universe, i.e.,
the red, yellow and blue vertices.  The  following $\soe$ formula expresses $\Col$.
\begin{eqnarray*} 
\Phi\sx{3-color} &\equiv& \exists R\, Y B\, \forall x y\, \Bigl( \bigl(R(x)\lor Y(x)\lor B(x)\bigr)
\; \land \; \bigl(E(x,y)\sra  \\ & & \qquad\lnot\bigl(R(x) \land R(y)\bigr)\, \land \,
\lnot\bigl(Y(x) \land Y(y)\bigr)\, \land \, \lnot\bigl(B(x) \land B(y)) \bigr) \Bigr)
\end{eqnarray*}
This is an example of Fagin's Theorem -- a completetly logical characterization of the 
complexity class NP.  A property is expressible in $\soe$ iff it is in NP:

\begin{theorem}{\bf (Fagin's Theorem)} {\rm \cite{Fagin}}\label{Fagin th}
$\np$ is equal to the set of decision problems expressible by existential, second-order
sentences, in symbols, $\np = \soe$.
\end{theorem}

 \section{First-Order Reductions}

 We now explain reductions appropriate for Descriptive Complexity.  
Here is a simplified definition.

\begin{definition}{({\bf First-Order Reduction})}\label{first-order query de}
 Let $\sigma$ and $\tau$ be vocabularies where $\tau = \angle{R_1^{a_1},\ldots,R_r^{a_r}}$, and let
$k$ be a fixed natural number.  A {\em $k$-ary first order mapping},
$I:\struc\sigma\rightarrow\struc\tau$ is given by $r$ first-order formulas of vocabulary
$\sigma$, $I = \angle{\phi_1, \ldots, \phi_r}$.  Given a structure $\A\in\struc\sigma$ with
universe $A$, $I$ maps $\A$ to
\[I(\A) \qeqdef (A^k,R_1^{I(\A)},\ldots,R_r^{I(\A)})\]
where the universe of $I(\A)$ is the set of $k$-tuples from the universe of $\A$, and the relation
$R_i^{I(\A)}$ is defined by the formula $\phi_i$ as follows.
\[ R_i^{I(\A)}\seqdef
\bigset{(\angle{b_1^1,\ldots,b_1^k},\ldots,\angle{b_{a_i}^1,\ldots,b_{a_i}^k})\in\abs{I(\A)}^{a_i}}{\A\models\phi_i(b_1^1,\ldots,b_{a_i}^k)}\]

We define a {\em first-order reduction} as a many-one reduction computed by a first-order mapping and
we write $A \foR B$ to mean that $A$ is reducible to $B$ via a first-order reduction.
\end{definition}

Let $\Sat$ be the set of satisfiable propositional formulas in conjunctive normal form, with at most
3 literals per clause.  (A literal is a variable or its negation.)  It follows from
Fagin's Theorem that $\Sat$ is NP complete viat first-order reductions \cite{book}.

To illustrate the concept of first-order reductions we present a reduction from $\Sat$ to $\Col$:

\begin{proposition}\label{Col for complete pr}
 $\Col$ is $\np$-complete via first-order reductions.
\end{proposition}

\begin{proof}
 We show that $\Sat\foR\Col$.  We are given a structure $\A$ encoding a 3-CNF formula and we must
 produce a 
 graph $f(\A)$ that is three colorable iff $\A\in\Sat$.  Let $n=\card{\A}$, so $\A$ is a boolean
 formula with at most $n$ variables and $n$ clauses.

 The construction of $f(\A)$ is shown in Figure \ref{color fi}.  Notice the triangle, with vertices
 labeled $T,F,R$.  Any three-coloring of the graph must color these three vertices distinct colors.
 We may assume without loss of generality that the colors used to color $T,F,R$ are true, false, and
 red, respectively.

 Graph $f(\A)$ also contains a ladder each rung of which is a variable $x_i$ and its negation
 $\ov{x_i}$.  Each of these is connected to $R$, meaning that  one of
 $x_i,\ov{x_i}$ must be colored true and the other false.

\begin{myfigure}\label{color fi}
\centerline{\includegraphics[width=\textwidth]{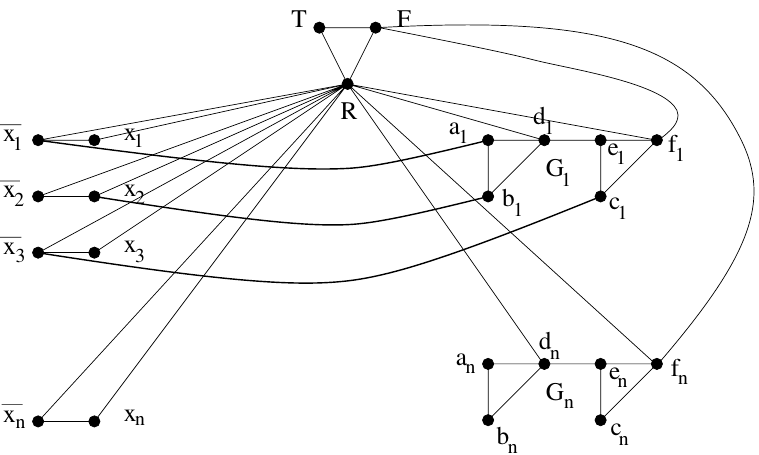}}
\ctrline{{\bf Figure \ref{color fi}: } $\Sat\foR\Col$; $G_1$ encodes clause $C_1 = (\ov{x_1}\lor
{x_2}\lor \ov{x_3})$}
\end{myfigure}

 Finally, for each clause $C_i =\ell_1\lor\ell_2\lor\ell_3$, $f(\A)$ contains the gadget $G_i$
 consisting of six vertices.  $G_i$ has three inputs $a_i,b_i,c_i$, connected to literals
 $\ell_1,\ell_2,\ell_3$, respectively, and it has one output, $f_i$. See Figure \ref{color fi} where
 gadget $G_1$ encodes clause $C_1 = \ov{x_1}\lor {x_2}\lor \ov{x_3}$.

 Observe that the triangle $a_1,b_1,d_1$ serves as an ``or''-gate in that $d_1$ may be colored true
 iff at least one of its inputs $\ov{x_1},x_2$ is colored true.  Similarly, output $f_1$ may be
 colored true iff at least one of $d_1$ and the third input, $\ov{x_3}$, is colored true.  Since
 $f_i$ is connected to both $F$ and $R$, $f_i$ can only be colored true.  It follows that a three
 coloring of the literals can be extended to color $G_i$ iff the corresponding truth assignment
 makes $C_i$ true.  Thus, $f(\A)\in\Col$ iff $\A\in\Sat$.

 The details of first-order reduction $f$ are easy to fill in.  $f(\A)$ consists of one triangle, a
ladder with $n$ rungs, and $n$ copies of the gadget.  The only dependency on the input $\A$ --- as
opposed to its size --- is that there is an edge from literal $\ell$ to input $j$ of gadget $G_i$
iff $\ell$ is the $j\th$ literal occurring in $C_i$.
\end{proof}

\subsection{Projections}

Projections are not-necessarily uniform reductions such that 
each bit of the output depends on at most one bit of the input. 

\begin{definition}\label{projection de} (Projection)  \cite{Valiant}\quad  A $k$-ary {\em
    projection}, $f$, is a 
  function from binary strings to binary strings, $f(w) = f_{\abs{w}}(w)$, determined by a family
 of functions, $f_n: \set{0,1}^n \rightarrow \set{0,1}^{n^k}$, having the following property.
 There is a function $t$ such 
that for each $n$, each $w \in \set{0,1}^n$ and each $j\leq
 n^k$, bit $j$ of $f_n(w)$ is one of the following, and the same choice is made for all $w\in
 \set{0,1}^n$:
  \begin{itemize}
    \item bit $j$ of $f_n(w)$ is 0 for all $w\in \set{0,1}^n$, or, \item bit $j$ of $f_n(w)$ is 1
    for all $w\in \set{0,1}^n$, or, \item bit $j$ of $f_n(w)$ is $w_{t(n,j)}$ for all $w\in
    \set{0,1}^n$, i.e. bit $t(n,j)$ of the input, $w$, or
  \item bit $j$ of $f_n(w)$ is $\lnot
    w_{t(n,j)}$ for all $w\in \set{0,1}^n$, i.e. the negation of bit $t(n,j)$ of the input, $w$.
  \end{itemize}

Equivalently, $k$-ary projections are the set of functions computed by a 
family of circuits where the $n$th circuit computes $f_n$, and consists of $n^k$ wires -- one for
each output bit -- each coming from a constant, 0 gate or a single input gate; with no computatiion
gates except at most one ``not'' gate per wire.
\end{definition}

As an example, the reduction $\Sat \foR \Col$ we gave in Proposition \ref{Col for complete pr} is a
projection.  On input $\A$, the first-order reduction produces a binary string encoding the
adjacency matrix of the graph $f(\A)$.  Each particular bit of this adjacency matrix is 1 if the
corresponding edge is present in $f(\A)$.  Look at Figure \ref{color fi} and note that all of the
drawn edges are the constant ``1'', all the non-drawn edges are the constant ``0'' with the
exception of the edges $(\ell,a_i), (\ell,b_i), (\ell,c_i)$ for literal $\ell$.  Each of these comes
from a single bit of the input, $\A$.  For example, the output bit $E(x_7,a_1)$, saying that
there is an edge from literal $x_7$ to vertex $a_1$ in $f(\A)$ is equal to the input bit of $\A$
which says whether $x_7$ is the first literal of $C_1$ in $\A$,  
Similarly, output bit $E(\ov{x_7},b_1)$ is equal to the input bit determining whether $\overline{x_7}$
is the second literal of $C_1$ in $\A$.

We now define first-order projections. These are simply projections that are also first-order
reductions.  We syntacticly restrict the formulas defining the reduction to force them to be
projections.

\begin{definition}\label{fop de}({\bf First-Order Projection} and {\bf Quantifier-Free Projection}): 
 Let $I =\angle{\phi_1,\ldots,\phi_r}$ be a $k$-ary first-order reduction from $S$ to $T$
 (Definition \ref{first-order query de}).  Suppose that for $i\geq 1$, $\phi_i$ satisfies the
 following {\em projection condition}.

 \vspace*{.1in}

 \myequation{projection eq}{\phi_i \equiv \alpha_1 \lor
   (\alpha_2\land \lambda_2)\lor\cdots\lor (\alpha_e\land \lambda_e)}

 \vspace*{.1in}

 \noindent where the $\alpha_j$'s are
 mutually exclusive numeric formulas, i.e., no input relations occur, and each $\lambda_j$ is a literal,
 i.e. an atomic formula $P(x_{j_1},\ldots x_{j_a})$ or its negation.

In this case, predicate $R_i(\angle{u_1,\ldots,u_k},\ldots,\angle{\ldots,u_{ka_i}})$ holds in
$I(\A)$ if $\alpha_1(\bar u)$ is true or if $\alpha_j(\bar u)$ is true for some $1<j\leq e$ and if
the corresponding literal $\lambda_j(\bar u)$ holds in $\A$.  Thus, each bit in the binary
representation of $I(\A)$ is determined by at most one bit in the binary representation of $\A$.  We
say that $ I$ is a {\em first-order projection} (fop).  Write $S\fop T$ to mean that $S$ is
reducible to $T$ via a first-order projection.

If $I$ is a fop whose defining formulas are quantifier-free, then $I$ is a {\em
quantifier-free projection} (qfp).  Write $S\qfp T$ to mean that $S$ is reducible to $T$ via a
qfp.\end{definition}

The key features of a fop $I$ are (1) $I$ is a projection (each bit of $I(\A)$ is determined
by at most one bit of $\A$) and (2) the first-order numeric formulas $\alpha_i$ in Equation
(\ref{projection eq}) tell us which bit of the input to copy and whether or not to negate it.

The difference between a qfp and a fop is that in a qfp, the numeric formulas $\alpha_i$ are
actually quantifier-free.  From Fagin's Theorem, we have that $\sat$ and $\Sat$ are complete via
qfps.
Furthermore, the reader should check\footnote{The best way to understand this is to write it out for
yourself.  For other examples of fops and qfps, see \cite{book}.} that the reduction we gave in
Proposition \ref{Col for complete pr} can be written as a qfp\footnote{In a paper exploring
  first-order and quantifier-free reductions, Elias Dahlhaus
  correctly proves using his definition that $\Col$ is NP complete via first-order reductions but
  {\bf not} via quantifier-free reductions  \cite{Dahlhaus}.  The crucial difference is that
  in his definitions, Dahlhaus did not allow numeric relations such as $\suc$.}.

\begin{proposition}\label{Col qfp complete pr}
 $\Sat\qfp\Col$ and thus $\Col$ is $\np$-complete via qfp's.
\end{proposition}

\section{Transitive Closure}

Let $\reach$ be the set of directed graphs having a path from $\mn$ to $\mx$.  Let $\reachd$ be the
subset of $\reach$ whose graphs have outdegree at most 1.  In \cite{capture}, I introduced a
transitive closure operator.  Let $\phi(x_1,\ldots ,x_k,x_1^\prime \ldots x_k^\prime )$ be a formula
of vocabulary $\tau$ with $2k$ free variables.  For any $\A \in \struc\tau$, $\phi$ expresses a
binary relation on $k$-tuples from $\A$,
\[\phi^\A \qe \bigset{(\ov{a},\ov{a'})}{(\A,\ov{a},\ov{a'}) \models \phi}\; .\]

The {\em transitive closure} operator, $\tc$, denotes the reflexive, transitive closure of such
relations.  Thus, for example, $\reach = \bigset{G}{G \models \tc(E)(\mn,\mx)}$.

I also introduced a deterministic version of transitive closure, $\dtc$.  Given a first order
relation, $\phi(\ov{x},\ov{y})$, define its deterministic reduct,
\[\phi_d(\ov{x},\ov{y})\quad \equiv\quad \phi(\ov{x},\ov{y})\,
\land\, (\forall \ov{z})( \lnot \phi(\ov{x},\ov{z}) \lor (\ov{y} = \ov{z}))
\]
That is, $\phi_d (\ov{x},\ov{y})$ is true just if the unique $\phi_d$ edge from $\ov{x}$ is to $\ov{y}$.
Now define:
\[(\dtc\, \phi) \quad \equiv \quad  (\tc\, \phi_d)\; .\]

Thus, for example,
\[\reachd = \bigset{G}{G \models \left(\dtc(E)(\mn,\mx)\land \mbox{outdegree}(G)\leq 1\right)}\; . \]

The operators $\tc$ and $\dtc$ capture the power of nondeterministic and deterministic logspace. (I
also introduced an alternating transitive closure operator, $\atc$, capturing the power of
$\aspace[\log n] = \ptime$ \cite{capture}.)

\begin{theorem}\label{dtc is l th}  {\rm \cite{capture}}  
  \[ \fo(\dtc)\se\logspace\;  \qquad \mbox{{\rm and }} \qquad \fo(\ptc) \se \nl
  \qquad \mbox{{\rm and }}  \qquad \fo(\atc) \se \ptime
  \; .\]
\end{theorem}

In the following, $\reacha$ is an alternating version of $\reach$.

\begin{theorem}\label{qfp complete th} {\rm \cite{capture}}
  \[
   \begin{array}{cccc}
    \reachd &\mbox{{\rm is qfp complete for}}& \logspace,& \quad \mbox{{\rm and}}\\
    \reach &\mbox{{\rm is qfp complete for}}& \nl,& \quad \mbox{{\rm and}}\\
   \reacha  &\mbox{{\rm is qfp complete for}}& \ptime.&\\
   \end{array}
   \]
\end{theorem}

Hartmanis advocated trying to prove that certain very weak reductions do not exist, as a way
to separate complexity classes. Here are four enticing examples. For example, to prove that $\ptime\ne\np$ it is necessary and sufficent to
show that there is no qfp from $\Col$ to $\reacha$.  The problem $\qsat$ is a
PSPACE-complete problem.

\begin{corollary}\label{lower bound co} {\rm \cite{capture}}  
  \[
  \begin{array}{ccc}
    \logspace = \np &\qLra&   \Col \qfp \reachd\\[2ex]
    \nl = \np       & \qLra&  \Col \qfp \reach\\[2ex]
    \ptime = \np       & \qLra&  \Col \qfp \reacha\\[2ex]
    \ptime = \pspace       & \qLra&  \qsat \qfp \reacha\\[2ex]
  \end{array}
  \]
\end{corollary}

\section{$\nspace[s(n)] \se \co\nspace[s(n)]$}

In the summer of 1987, I proved the following theorem, which was also independently proved by
R\'obert Szelepcs\'enyi.

\begin{theorem}\label{space th} {\rm \cite{space,Robert}} $\;$
  For any $s(n) \geq \log n$,
\[\nspace[s(n)] \qe \co\nspace[s(n)] \; .\]
\end{theorem}

This had been a long-standing open problem \cite{Kuroda}.  Everyone I discussed it with had
conjectured the negation of 
Theorem \ref{space th}.  Nondeterministic space corresponds to an existential search, and it seemed
natural --- based on our intuition concerning existential quantification --- that this class would
not be closed under complementation.

I was looking at the relationship between $\nl$ and $\co\nl$ because of a surprising new result,
namely that the alternating logspace hierarchy collapses at the second level, i.e., $\Sigma_2^L =
\Pi_2^L$ \cite{JKL}.  Like everyone else, I still believed that $\nl \ne \co\nl$ and I knew that by
Theorem \ref{qfp complete th}, this would be true iff $\ov{\reach} \not\qfp \reach$.

So, I tried to prove that $\ov{\reach} \not\qfp \reach$.  My proof idea, was, suppose there is a
reduction $\ov{\reach} \qfp \reach$, what would it look like?  I knew from  the definition of qfps
(Definition \ref{fop de}) and the proof of Theorem \ref{qfp complete  th} that  there is very little
flexibility concerning what such a qfp would look like. 
Trying to show exactly what it would have to look like if it existed, I built it:

\begin{theorem}\label{reachbar th}
\quad $\ov{\reach} \qfp \reach$.
\end{theorem}

The proof of Theorem \ref{reachbar th} is easily translated to a counting argument on Turing
Machines.  The night I proved this, I wrote it up that way because it was more accessible to general
theoreticians, who are not familiar with qfps.  The following is a version of the proof.
I think it is a sufficient guide for any reader who wants to explicitly write
out the qfp of Theorem \ref{reachbar th}.

\begin{theorem}\label{complement th} {\rm \cite{book}} $\;$ For finite, ordered structures,
\[ \quad \fo(\ptc) \quad =\quad \fo(\tc)\; .\]
\end{theorem}

\begin{sproof}
It suffices to show that the relation $\lnot\tc (E)(\mn,\mx)$ meaning that there is no path from
$\mn$ to $\mx$, is expressible in $\fo(\ptc)$.

To do this, we count the number of reachable vertices.  Consider an input graph, $G$.  As usual, we
consider the elements of $G$ as numbers as well as vertices.  In one setting, as distances, we think
of these numbers as ranging from $0$ to $n-1$.  In another setting, as counts of the number of
reachable vertices, we have numbers ranging from 1 to $n$.  Writing these two sets of numbers as
numbers rather than as vertices makes our notation simpler to understand.

Define $n_d$ to be the number of vertices in $G$ that are reachable from $\mn$ in a path of length
at most $d$.  Given number $n_d$, we show how to compute number $n_{d+1}$.  As a first step, we show
that $n_d$ allows us to say in $\fo(\ptc)$ that there is {\em not} a path of length at most $d$ from
$\mn$ to a given vertex.

\begin{claim}\label{ndist cl}
The following formulas are expressible in $\fo(\ptc)$:
\begin{enumerate}
\item $\dist(x,d)$, meaning that there is a path of length at most $d$ from $\mn$ to $x$ \item
$\ndist(x,d;m)$, which --- when $m=n_d$ --- means that there is no path of length at most $d$ from
$\mn$ to $x$.
\end{enumerate}
\end{claim}

\begin{sproof}
There is no trouble writing $\dist(x,d)$ positively,
\begin{eqnarray*}
\dist(x,d) &\equiv& \tc(\alpha)(\mn,0,x,d), \quad \mbox{where} \\[1ex] \alpha(a,i,b,j) &\equiv&
(E(a,b) \lor a=b) \;\land\; \suc(i,j)
\end{eqnarray*}

We write the formula, $\ndist(x,d;m)\in \fo(\ptc)$ to mean the following:
\[\ndist(x,d;m) \sE \mbox{(There are at least $m$ vertices $v$)}(v\ne x
\sland \dist(x,d))\; .\] It follows that when $m= n_d$, we have that $\ndist(x,d;m)$ is equivalent to
$\lnot \dist(x,d)$.

Define edge relation $\beta$ on pairs of vertices as follows,
\begin{eqnarray*}
\beta(v,c,v',c') &\equiv& \mn \ne x \;\land\; \suc(v,v')\\ &&\land\quad (c=c' \;\lor\; (\suc(c,c')
\land \dist(v',d) \land v'\ne x))\; .
\end{eqnarray*}
Suppose that $c$ is the number of vertices --- not including $x$ --- that are at most $v$ and
reachable from $0$ in at most $d$ steps.  Then we can take a $\beta$-step from $\angle{v,c}$ to
$\angle{v+1,c}$ guessing that $v+1$ is not reachable from $0$ in $d$ steps; or, we may take a
$\beta$-step to $\angle{v+1,c+1}$ {\bf if we prove} that $v+1$ is not equal to $x$ and is reachable
from $0$ in $d$ steps.

Thus, there is a path from $\angle{0,1}$ to $\angle{v,c}$ iff there are at least $c$ vertices not
equal to $x$ and at most $v$ such that $\dist(v,d)$:
\[ \tc(\beta)(0,1,v,c) \qLra  c\leq \left|\bigset{w}{w\leq v \sland
\dist(w,d)}\right|\; . \] $\ndist$ can now be defined as follows,

\labeleqn{\Box}{\ndist(x,d;m) \qE \tc(\beta)(\mn,1,\mx,m)\; .}
\end{sproof}

Using Claim \ref{ndist cl}, we now define the relation $\delta(d,m,d',m')$ so that if $m= n_d$, then
$m' = n_{d+1}$.  We simply cycle through all the vertices, counting how many of them are reachable
in $d+1$ steps:
\begin{eqnarray*}
\delta(d,m,d',m') &\equiv& \suc(d,d') \;\land\; \tc(\gamma)(\mn,1,\mx,m')\\[1ex] \gamma(v,c,v',c')
&\equiv& \suc(v,v') \sland \Bigl([\suc(c,c') \sland \dist(v',d+1)] \quad \lor\\ &&\;[c=c' \land
(\forall z)(\ndist(z,d;m) \lor (z\ne v' \land \lnot E(z,v')))]\Bigr)\\
\end{eqnarray*}
It follows that formula $\tc(\delta)(0,1,n-1,m)$ holds iff $m = n_{n-1}$ is the number of vertices
in $G$ that are reachable from $\mn$.  As claimed, using this $m$, we can express the
nonexistence of a path,  \labeleqn{\Box}{\lnot \tc(E)(\mn,x) \quad\equiv\quad (\exists
m)(\tc(\delta)(0,1,n-1,m) \,\land\, \ndist(x,n-1;m))}
\end{sproof}

\section{Concluding Comments}

Inspired by my adviser, and thinking about  Theorem \ref{reachbar th} and Corollary \ref{lower bound
  co}, I believe that fops and qfps have great
potential for answering complexity questions.  For a certain qfp to exist between given
problems, there is very little flexibility concerning what it would have to look like.  And very
little flexibilty in what you are searching for can be a good thing.

As I was leaving Cornell after my Ph.D. defense, Hartmanis said to me, ``You might wonder how
you will find time to do research as a professor when you have to prepare classes, teach, hold
office hours, grade, write letters of recommendation for your students, serve on university and
department committees, interview and hire new faculty, and so on.''   I asked how
he did it and he said, ``I take long showers''  \cite{tenureTrack}.

It is common experience that creative works come together at surprising times.  For Hartmanis, it
was showers.  For me it's walking with dogs.  I completed the qfp reducing
$\ov{\reach}$ to $\reach$, and thus the proof that nondeterministic space is closed under
complementation, while I was giving my Irish Setter, Molly, her evening walk.

\begin{center}
  \includegraphics[width = .86\textwidth]{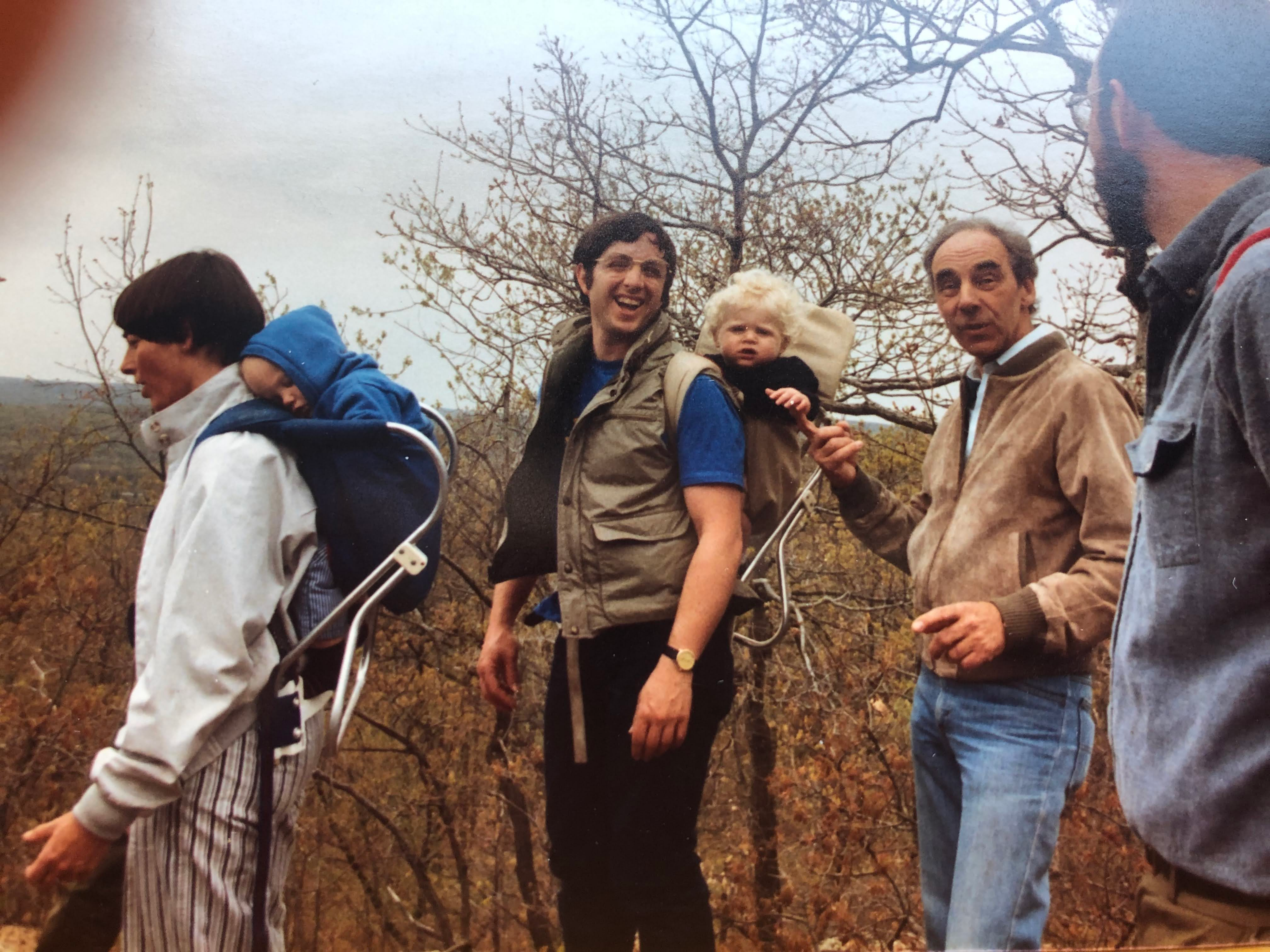}

{\small Fran, Geoff Kozen, Neil, Daniel Immerman, Juris Hartmanis, Scott Drysdale}
\end{center}

\bfhead{Acknowledgements:} The above photograph was taken by Susan Landau. Much thanks to Eric
Allender and Phokion Koklaitis for very helpful comments on a previous draft.


\begin{thebibliography}{WWW99}
\bibitem[AHU74]{AHU}  A.V. Aho, J.E. Hopcroft and J.D. Ullman, {\it The
    Design and Analysis of Computer Algorithms,} Addison-
    Wesley (1974).
\bibitem[Bul17]{Bulatov}
Andrei Bulatov, ``A dichotomy theorem for nonuniform csps, `` \focs (2017),
319-330.
\bibitem[Coo71]{Cook} Stephen Cook, ``The Complexity of Theorem Proving Procedures,'' {\it
  Proc. Third Annual ACM STOC Symp.}  (1971), 151-158.
\bibitem[Dah83]{Dahlhaus} Elias Dahlhaus, ``Reduction to NP-Complete Problems by Interpretations''
  in: B\"{o}öger, E., Hasenjaeger, G., R\"{o}dding, D., eds., {\bf Logic and Machines: Decision
    Problems and Complexity} (1983),  Lecture Notes in Computer Science, vol 171. Springer.
\bibitem[Fag74]{Fagin} Ron Fagin,
``Generalized First-Order Spectra and Polynomial-Time Recognizable Sets,'' in {\it Complexity of
  Computation,} (ed. R. Karp), {\it SIAM-AMS Proc. 7,} 1974, 27-41.
\bibitem[FV99]{FV99} Thomas Feder
and Moshe Vardi, The computational structure of monotone monadic SNP and constraint satisfaction: a
study through Datalog and group theory,'' \sicomp 28 (1999) 57 - 104.
\bibitem[HIM78]{HIM}Juris Hartmanis, Neil Immerman,
  and Stephen Mahaney, ``One-Way Log Tape Reductions,'' \focs (1978), 65-72.
\bibitem[Imm87]{capture}
Neil Immerman, ``Languages That Capture Complexity Classes,'' {\it SIAM J. Comput.} {\bf 16}, No. 4
(1987), 760-778.
\bibitem[Imm88]{space} Neil Immerman, ``Nondeterministic Space is Closed Under
Complementation,'' {\it SIAM J. Comput.} 17(5) (1988), 935-938.  Also appeared in {\it Third
  Structure in Complexity Theory Conf.} (1988), 112-115.
\bibitem[Imm99]{book}Neil Immerman, {\it
Descriptive Complexity}, 1999, Springer Graduate Texts in Computer Science, New
  York.
\bibitem[JKL89]{JKL} Birgit Jenner, Bernd Kirsig, and Klaus-J\"orn Lange, ``The Logarithmic
Hierarchy Collapses: $A\Sigma_2^L = A\Pi_2^L$,'' {\it Information and Computation} 80 (1989),
269-288.
\bibitem[Jon73]{Jones} Neil Jones, ``Reducibility Among Combinatorial Problems in Log $n$
Space,'' {\it Proc. Seventh Annual Princeton Conf. Info. Sci. and Systems} (1973),
547-551.
\bibitem[Kar72]{Karp} Richard Karp, ``Reducibility Among Combinatorial Problems,'' in {\it
Complexity of Computations,} R.E.Miller and J.W.Thatcher, eds. (1972), Plenum Press,
  85-104.
\bibitem[Kur64]{Kuroda} S.Y. Kuroda, ``Classes of Languages and Linear-Bounded
    Au\-tom\-ata,'' {\it Information and Control} {\bf 7} (1964), 207-233.
\bibitem[Lad75]{Ladner} Richard Ladner, "On the structure of polynomial time reducibility,"
  \jacm 2(1) (1975), 155-171.
\bibitem[Lan91]{tenureTrack} Susan Landau, ``Tenure Track, Mommy Track'' {\it Association
  for Women in Mathematics Newsletter,} May-June, 1991.
\bibitem[Sch78]{Schaefer}Thomas Schaefer, ``The Complexity of Satisfiability Problems,"
  \stoc (1978), 216 - 226.
\bibitem[Sip12]{Sipser} Michael Sipser, {\it Introduction to the Theory
    of Computation, Third Edition}, Cengage Learning, (2012).PWS, 1997.
\bibitem[Sze88]{Robert} R\'obert Szelepcs\'enyi, ``The Method of Forced
Enumeration for Nondeterministic Automata,'' {\it Acta Informatica} {\bf 26} (1988),
279-284.
\bibitem[Val82]{Valiant} Leslie Valiant, ``Re\-duc\-ibil\-ity By Algebraic Projections,''
{\it L'En\-seigne\-ment math\'e\-ma\-tique,} {\bf T. XXVIII}, 3-4 (1982),
253-268.
\bibitem[Zhu20]{Zhuk} Dmitriy Zhuk, ``A Proof of the CSP Dichotomy Conjecture'', \jacm
67(5), (2020), 1-78, \url{https://doi.org/10.1145/3402029}
\end{thebibliography}
\end{document}